\documentclass{article}

\usepackage{amsmath}
\usepackage{amssymb}
\usepackage{amsfonts}
\usepackage{graphicx}
\usepackage{subfigure}

\newcommand{\ket}[1]{\vert {#1} \rangle}  
\newcommand{\bra}[1]{\langle {#1} \vert}

\title{Holistic type extension for classical propositional logic in quantum computation}
\author{{\sc H. Freytes}, {\sc R. Giuntini}, {\sc G. Sergioli}}

\date{{\small
Dipartimento di Filosofia, Universit\'{a} di Cagliari, Via Is Mirrionis 1, 09123, Cagliari-Italia}}

\begin{document}

\maketitle

\begin{abstract}
An holistic extension of classical propositional logic is introduced in the framework of quantum computation with mixed states. The concepts of tautology and contradiction are investigated in this extensions. A special family of quantum states are investigated as particular cases of ``holistic" contradiction.
\end{abstract}

\begin{small}
\end{small}


\newtheorem{theo}{Theorem}[section]
\newtheorem{definition}[theo]{Definition}
\newtheorem{lem}[theo]{Lemma}
\newtheorem{prop}[theo]{Proposition}
\newtheorem{coro}[theo]{Corollary}
\newtheorem{exam}[theo]{Example}
\newtheorem{rema}[theo]{Remark}{\hspace*{4mm}}
\newtheorem{example}[theo]{Example}
\newcommand{\proof}{\noindent {\em Proof:\/}{\hspace*{4mm}}}
\newcommand{\qed}{\hfill$\Box$}

\section*{Introduction}
In recent years, an increasing interest on logical systems related to quantum mechanics has arisen. Most of these systems are not strictly related to the standard quantum logic but they are focused on concrete problems based on some particular situation associated to quantun systems. It motivates foundational problems also related to the approaches to knowledge \cite{BDa,D1,D2}. In our case, we study an holistic type extension of the classical propositional logic in the context of quantum computing. In this extension the notion of classical contaddition exhibits an interesting behaviour that suggests a sort of relation with paraconsistent logic associated to quantum superposition \cite{DaD}.
 
The concept of quantum computing was introduced at the beginning of the 1980s by Richard Feynman. One of the central issues he posed,  
was the difficulty of efficiently simulating the evolution of a quantum system on a classical computer. He pointed out the computational benefits that arise by employing quantum systems in place of classical ones. With this aim he proposed a new kind of computer: a {\it quantum computer} \cite{FEY}. It was not conceived as a Turing machine, but as a different kind of machine able to simulate any quantum system, including the physical world. Quantum computing can simulate all computations which can be performed by classical systems. However, one of the main advantages of quantum computation and quantum algorithms is that they can speed up computations.

The idea of quantum computation was only theoretically interesting until Deutsch introducted the concept of universal quantum computer \cite{DEU} and Shor discovered an algorithm able to factorize large numbers in a polynomial time \cite{SH}.  After that, several new researches on quantum computation were started.

Quantum computation is motivated by the fact that quantum systems make possible new interesting forms of computational and communication processes. In fact, quantum computation can be seen as an extension of classical computation where new primitive information resources are introduced. One of the main ingredients of such an extension is the notion of quantum bit (\emph{qubit} for short) which is the quantum computational counterpart of the classical bit. On this basis, new form of computational processes are developed in order to operate with these new information resources.

In classical computation, information is encoded by a sequence of bits. A bit is viewed as a kind of physical object which can assume one of two distinct  classical states, represented by the binary numbers $0$ or $1$. Bits are manipulated via an ensemble of logical gates like {\it NOT, OR, AND}, etc, arranged in circuits giving out the result of a calculation.

Standard quantum computing is based on quantum systems described by finite dimensional Hilbert spaces, starting from  ${\mathbb{C}}^2$, that is the
two-dimensional space where any qubit lives. A qubit is represented by a unit vector
in ${\mathbb{C}}^2$, while {\it $n$-qubits} (where $n$ is a positive integer) are represented by unit vectors in
${\mathbb{C}}^{2^n}$. Similarly to the classical case, we can introduce and study the behavior of a certain number of {\it quantum
logical gates} (hereafter \emph{quantum gates} for short) operating on $n$-qubits. As in the classical case, a quantum circuit is identified with an appropriate composition of quantum gates. These gates
are mathematically represented by unitary operators acting on pure states of an appropriate  Hilbert space and thus they only represent reversible processes. But for many reasons this restriction to unitary operators is unduly. In fact, a quantum system is rarely in a pure state. This may be
caused, for example, by the incomplete efficiency in the preparation procedure and also by manipulations on the system as measurements over pure states; in both cases we are faced with \emph{statistical mixtures}. Such restriction motivated the study of more general models for quantum computational processes, where pure states and unitary operators are replaced by density operators and quantum operations, respectively. This approach to quantum computing, where not only reversible transformations are taken into account, is called {\it quantum computation with mixed states}. 

In this powerful model, combinational structures associated to a set of quantum gates induce new forms of quantum logical systems that play a similar role to Boolean algebras with respect to digital circuits. We focus our attention on the combinational structure of quantum circuits built from a particular quantum gate: the \emph{Toffoli quantum gate}.

The study of the combinational logic underlying Toffoli quantum gate is interesting for several reasons. One of these is related to the universality of quantum gates. A finite set of quantum gates is said to be {\it universal} if and only if any other unitary quantum gate can be expressed as a finite sequence of gates of such a set. Mathematically, by a simple cardinality argument, we can see that such a set can not exist. 
To investigate the universality problem, we can refer to a 
weaker notion of approximate universality.
A finite set of quantum gates is said to be approximate universal if and only if any quantum gate can be approximated by a finite  sequence of gates of the set. 
Quantum computation only needs two gates to be approximate universal: the Toffoli and the Hadamard gates \cite{SHI}.
The fact that Toffoli and Hadamard gates define an approximate universal set is of a particular interest from a foundational point of view: as is  well known, the Toffoli gate can exactly perform all classical reversible computation. Consequently, the Hadamard gate is all we need to add to classical computation in order to achieve the full quantum computational  power. In this perspective, the combinational logic associated to Toffoli gate is the logic of classical reversible processes in quantum computation. 

Another reason that makes interesting the logic of Toffoli gate, is its connection with fuzzy logic. Indeed, from a probabilistic point of view, Toffoli gate behaves as the conjunction of {\it Product logic} \cite{CT}. This logical system is related to the so called {\it fuzzy logic of continuous $t$-norms} introduced by H\'ajek in the second part of 90's \cite{HAJ}. Further, Product logic is also related to game theory applied to classical communication with feedback \cite{MONT1,PELC}. It suggests potential implementations of Toffoli gate in the quantum error correction theory.

The aim of this paper is to study two extensions of classical logic arisen from the Toffoli gate: one comes from the fuzzy behavior of the Toffoli gate mentioned above and the other emerges from the holistic nature of bipartite quantum systems.

The paper is structured as follows: the first two sections provide all the necessary ingredients to make the article self-contained. More precisely, in Section 1 we introduce some basic notions concerning non-separability and bipartite quantum systems. In Section 2, we briefly describe the model of quantum computation with mixed states. In Section 3 we introduce a general logical framework associated to quantum circuits. This new form of quantum logic is compared to the standard quantum logic based on the closed subspaces of the Hilbert space (\emph{Hilbert lattices}). Section 4 is devoted to  studying the fuzzy extension arising from Toffoli gate. This extension will be defined by means of two particular instances of Toffoli gate: 
$\mathbb{AND}$ and $\mathbb{NOT}$. In Section 5, an holistic type extension for  classical logic is investigated. This extension is motivated by the application of $\mathbb{AND}$ gate on non-separable states. We also study the notion of contradiction in the holistic extension of classical logic and Werner states are introduced as particular cases of these contradictions. 

Some arguments and possible open discussions are briefly introduced as conclusive remarks.

\section{Bipartite quantum systems.} \label{HOL}

The notion of {\it state of a physical system} is familiar from its use in classical mechanics, where it is linked to the initial
conditions (the initial values of position and momenta) which determine the solutions of the equations of motion for the system. For
any value of time, the state is represented by a point in the phase space. In classical physics, compound systems can be decomposed into their subsystems. Conversely, individual systems can be combined to give overall composite systems. In this way, a classical global system is completely described in terms of the states of its subsystems and their mutual dynamic interactions. In other words, classical physics follows a separability principle that can be schematically expressed as:

\begin{itemize}
\item[] {\it Separability Principle}: The states of any space-time separated subsystems $S_1, S_2, \ldots, S_n$ of a compound system $S$ are individually well defined and the states of the compound system are wholly and completely determined by them and by their physical interactions,  including their space-time relations \cite{HAY, HOW}.

\end{itemize}

From a mathematical point of view, the separability condition of classical systems comes from the fact that states of compound systems are represented as direct sum of the states of their subsystems. 

In quantum mechanics the  description of the state  is substantially modified. The state of a quantum system embodies the specific history prior to the instant which the state refers to. Before giving the definition of quantum state, we introduce the concept of {\it maximal quantum test}. Suppose that we want to observe the property of a quantum system that can possibly take $n$ different values. If the test allows us to distinguish between $n$ possibilities, we say that it is \emph{maximal}. A $n$-outcome measurement of this property implements a maximal test. A test that gives only partial information about the measurement property is said to be \emph{partial}. If a quantum system is prepared in such a way that one can arrange a maximal test  that yields with certainty a particular outcome then we say that the quantum system is in a {\it pure state}. A pure state is described by a unit vector in a Hilbert space and it is denoted by $\vert \varphi \rangle$ in Dirac notation. If the maximal test for a pure state has $n$ possible outcomes,  the state is described by a vector $\vert \varphi \rangle$ in a $n$-dimensional Hilbert space. Any orthonormal basis represents a realizable maximal test. Suppose that we have an ensemble of similarly prepared systems and we test the values of
different measurable quantities (like spin, etc...). In general we postulate that, for an arbitrary ensemble state, it is always possible to
devise a test that yields the $n$ outcomes corresponding to an orthonormal basis with definite probability. If the system is prepared in the state $\vert \varphi \rangle$ and a maximal test corresponding to a basis $\{\vert e_1 \rangle \ldots \vert e_n \rangle\}$ is performed, then the probability that the outcome corresponds to $\vert e_i \rangle$ is given by $p_i(\vert \varphi \rangle) = \vert \langle e_i \vert \varphi \rangle \vert^2 $. When a quantum system is not in a pure state, quantum states are represented by probability distibutions of pure states, giving rise to the so called \emph{mixed states}.
Mixed states are represented by {\it density operators} in a Hilbert space, i.e. positive, self-adjoint trace class operators with trace equal to one. In terms of density operators, a pure state $\vert \psi \rangle$ is represented in Dirac notation as a matrix product $\vert \psi \rangle \langle \psi \vert$. Thus, with respect to pure states, density operators provide a more general description of quantum states. 

In quantum mechanics a compound system is represented as a tensor product of Hilbert spaces, each of them representing the individual parts of the system. Unlike classical physics, standard quantum mechanics systematically violates the above separability principle. From a mathematical point of view, the origin of this difference arises from the tensor product structure related to Hilbert spaces and from the superposition principle \cite{AERTS, DIECKS}. More precisely, if $\rho_1$ and $\rho_2$ are two density operators in the Hilbert spaces ${\cal H}_1$ and  ${\cal H}_2$ respectively, the state of the compound system is represented by $\rho = \rho_1 \otimes \rho_2$ in  ${\cal H}_1 \otimes {\cal H}_2$. But not all density operators on ${\cal H}_1 \otimes {\cal H}_2$ are expressible in this form. This non-factorizability of quantum states is related to the fact that the direct sum of ${\cal H}_1$ and ${\cal H}_2$ is a proper subset of ${\cal H}_1 \otimes {\cal H}_2$. This behavior may be considered as the mathematical root of the holistic feature of quantum mechanics. In fact, there exist properties of quantum systems that characterize the whole system but that are not reducible to the local properties of its parts. It should be noticed that the notion of tensor product motivates a different description of event structure for compound quantum systems \cite{DVU, FOU} with respect to the classical event structure.

In what follows we provide a formal description of these holistic features based on \emph{generalized Pauli matrices}. This approach turns out to be very useful to describe a holistic extension of classical logic in the quantum computation context.

Due to the fact that the Pauli matrices 
 \begin{equation*}
  \sigma_1 =
  \left[\begin{array}{cc}
    0 & 1 \\
    1 & 0
  \end{array}\right],
  \quad
  \sigma_2 =
  \left[\begin{array}{cc}
    0 & -i \\
    i & 0
  \end{array}\right],
  \quad
  \sigma_3 =
  \left[\begin{array}{cc}
    1 & 0 \\
    0 & -1
  \end{array}\right]
\end{equation*}
and $I$ are a basis for the set of operators over ${\mathbb{C}}^2$, an arbitrary density operator $\rho$ over ${\mathbb{C}}^2$ may be represented
as $$\rho=\frac{1}{2}(I+s_1\sigma_1+s_2\sigma_2+s_3\sigma_3)$$ where $s_1,s_2$ and $s_3$ are three real numbers such $s_1^2+s_2^2+s_3^2\le 1$. The triple $(s_1,s_2,s_3)$ represents the point of the Bloch sphere that is uniquely associated to $\rho$. A similar canonical representation can be obtained for any $n$-dimensional Hilbert space by using the notion of generalized Pauli-matrices.

\begin{definition}
{\rm Let $\mathcal H$ be a $n$-dimensional Hilbert space and $\{\ket{\psi_1},\ldots,\ket{\psi_n}\}$ be the canonical othonormal basis of $\mathcal H$. Let $k$ and $j$ be two natural numbers such that: $1\le k < j \le n$. Then, the {\it generalized Pauli-matrices} are defined as follows: 
 $$^{(n)}\sigma_1^{[k,j]}=
\ket{\psi_j}\bra{\psi_k}+
\ket{\psi_k}\bra{\psi_j}$$
 $$^{(n)}\sigma_2^{[k,j]}=
i(\ket{\psi_j}\bra{\psi_k}-\ket{\psi_k}\bra{\psi_j})$$
and for $1\le k \le n-1$
$$^{(n)}\sigma_3^{[k]}= \sqrt{\frac{2}{k(k+1)}}(\ket{\psi_1}\bra{\psi_1}+\cdots+\ket{\psi_k}\bra{\psi_k}-k\ket{\psi_{k+1}}\bra{\psi_{k+1}}).$$
}
\end{definition}

If $\mathcal H=\mathbb C^2$ one immediately obtains:  $^{(2)}\sigma_1^{[1,2]}= \sigma_1$, $^{(2)}\sigma_2^{[1,2]}= \sigma_2$ and $^{(2)}\sigma_3^{[1]}= \sigma_3.$

Let $\rho$ be a density operator of the $n$-dimensional Hilbert space ${\cal H}$. For any $j$, where $1\leq j \leq n^2-1$, let $$s_j(\rho) = tr(\rho \sigma_j) .$$ The sequence $\langle s_1(\rho)\ldots s_{n^2-1}(\rho) \rangle$ is called the {\it generalized Bloch vector} associated to $\rho$, in view of the following well known result {\rm \cite{SM}}:
let $\rho$ be a density operator of the $n$-dimensional Hilbert space ${\cal H}$ and let $\sigma_j$ be the generalized $n$-dimensional Pauli matrices. Then $\rho$ can be canonically represented as follows:
\begin{equation}\label{BLOCHVECT}
\rho = \frac{1}{n}I^{(n)} + \frac{1}{2}\sum_{j=1}^{n^2-1}s_j(\rho)\sigma_j
\end{equation}
where $I^{(n)}$ is the $n\times n$ identity matrix.

A kind of converse of the above result reads: a matrix $\rho$ having the form $\rho = \frac{1}{n}I^{(n)} + \frac{1}{2}\sum_{j=1}^{n^2-1}s_j(\rho)\sigma_j $ is a density operator if an only if its eigenvalues are non-negative.
By using generalized Pauli matrices, it will be possible to formally describe a notion of holism for bipartite states. In fact, by following the Schlienz-Mahler decomposition \cite{SM}, we show as any quantum bipartite state can be expressed as a sum of a factorizable state plus another quantity that represents a kind of holistic component.

Let us consider the Hilbert space ${\cal H} = {\cal H}_a \otimes {\cal H}_b$. For any density operator $\rho$ on $\cal H$, we denote by 
$\rho_a$ the partial trace of $\rho$ with respect to the system ${\cal H}_b$ (i.e. $\rho_a = tr_{{\cal H}_b}(\rho)$) and by $\rho_b$ the partial trace of $\rho$ with respect to the system ${\cal H}_a$ (i.e. $\rho_b = tr_{{\cal H}_a}(\rho)$). For the next  developments it is useful to recall the following technical result: 

let $\rho $ be a density operator in a $n$-dimensional Hilbert space ${\cal H} = {\cal H}_a \otimes {\cal H}_b$ where $dim({\cal H}_a) = m$ and $dim({\cal H}_b) = k$.  If we divide $\rho$ in $m \times m$ blocks $B_{i,j}$, each of them is a $k$-square matrix, then:

\begin{eqnarray}\label{TRACE}
\rho_a & = &tr_{{\cal H}_b}(\rho)  =\left[\begin{array}{cccc}
      tr B_{1,1} & tr B_{1,2} & \ldots & tr B_{1,m}  \\
      tr B_{2,1} & tr B_{2,2} & \ldots & tr B_{2,m}  \\
      \vdots & \vdots & \vdots & \vdots  \\
      tr B_{m,1} & tr B_{m,2} & \ldots & tr B_{m,m}   \\
          \end{array}\right] \\
\rho_b & = &tr_{{\cal H}_a}(\rho) =\sum_{i=1}^m B_{i,i}.
\end{eqnarray}

\begin{definition}
{\rm Let $\rho $ be a density operator in a Hilbert space ${\cal H}_m \otimes {\cal H}_k$ such that $dim({\cal H}_m) = m$ and $dim({\cal H}_k) = k$. Then $\rho$ is said to be {\it $(m,k)$-factorizable} iff $\rho = \rho_m \otimes \rho_k$ where $\rho_m$ is a density operator in ${\cal H}_m$ and $\rho_k$ is a density operator in ${\cal H}_k$.   
}
\end{definition}

It is well known that, if $\rho$ is $(m,k)$-factorizable as $\rho = \rho_m \otimes \rho_k$, this factorization is unique and $\rho_m$ and $\rho_k$ correspond to the reduced states of $\rho$ on ${\cal H}_m$ and ${\cal H}_k$, respectively \cite{VON55}.\\

Suppose that ${\cal H} = {\cal H}_a\otimes {\cal H}_b$ where $dim({\cal H}_a) = m$ and $dim({\cal H}_b) = k$. Let us consider  the generalized Pauli matrices $\sigma_1^a,\ldots, \sigma_{m^2-1}^a $ and $\sigma_1^b,\ldots, \sigma_{k^2-1}^b $ arising from  ${\cal H}_a$ and ${\cal H}_b$, respectively. 

If we define the following coefficients: $$M_{j,l}(\rho) = tr(\rho [\sigma_j^a \otimes \sigma_l^b]) - tr(\rho [\sigma_j^a \otimes I^{(k)}])tr(\rho [I^{(m)}\otimes\sigma_l^b])$$ and if we consider the matrix ${\bf M}(\rho)$ defined as $${\bf M}(\rho)=\frac{1}{4} \sum_{j=1}^{m^2-1}\sum_{l=1}^{k^2-1} M_{j,l}(\rho)(\sigma_j^a \otimes \sigma_l^b)$$ then  ${\bf M}(\rho)$ represents the ``additional component" of $\rho$ when $\rho$ is not a factorized state.
In this way, if $\rho$ is a density operator in ${\cal H} = {\cal H}_a\otimes {\cal H}_b$, then
\begin{equation}\label{decomposition}
\rho = \rho_a \otimes \rho_b + {\bf M}(\rho).
\end{equation} 

The above result gives a formal representation of the instance of holism mentioned at the beginning of the section. In fact, a state $\rho$ in  
${\cal H}_a\otimes {\cal H}_b$ does not only depend on its reduced states $\rho_a$ and $\rho_b$, but also the summand ${\bf M}(\rho)$ is involved. Let us notice that ${\bf M}(\rho)$ is not a density operator and then it does not represent a physical state. We refer to ${\bf M}(\rho)$ as the \emph{holisitc component} of $\rho$.

\section{Quantum computation with mixed states}\label{QLMS}

As anticipated in the Introduction, we now provide some basic notions of quantum computing. 
In quantum computation, information is elaborated and processed by
means of quantum systems. 
A \textit{quantum bit} or \textit{qubit}, the fundamental concept of
quantum computation, is a pure state in the Hilbert space
${\mathbb{C}}^2$. The standard orthonormal basis $\{ \vert 0 \rangle ,
\vert 1 \rangle \}$ of ${\mathbb{C}}^2$, where $\vert 0 \rangle = (1,0)^\dagger$
and $\vert 1 \rangle = (0,1)^\dagger$, is generally called
\textit{logical basis}. This name refers to the fact that the logical truth is
related to $\vert 1 \rangle$ and the falsity to $\vert 0 \rangle$. Thus,
pure states $\vert \psi \rangle$ in ${\mathbb{C}}^2$ are coherent
superpositions of the basis vectors $\vert
\psi \rangle = c_0\vert 0 \rangle + c_1 \vert 1 \rangle$ where $c_0$ and $c_1$ are complex numbers such that $\vert
c_0 \vert^2 + \vert c_1 \vert^2 = 1$. 
Recalling the Born rule, any qubit $\vert \psi \rangle = c_0\vert 0
\rangle + c_1 \vert 1 \rangle$ may be regarded as a piece of
information, where the number $\vert c_0 \vert ^2$  corresponds to
the probability-value of the information described by the basic
state $\vert 0 \rangle$; while $\vert c_1 \vert ^2$  corresponds to
the probability-value of the information described by the basic
state $\vert 1 \rangle$. The two basis-elements $\vert 0 \rangle$
and $\vert 1 \rangle$ are usually taken as the encoding of the classical
bit-values $0$ and $1$, respectively. In this way, the qubit probability value we are interested on, is $p(\vert \psi
\rangle) = \vert c_1 \vert ^2$ that is related to the basis vector
associated with truth.

Quantum states considered in quantum computation, live in the
tensor product ${\otimes^n} {\mathbb{C}}^2 = {\mathbb{C}}^2 \otimes
{\mathbb{C}}^2 \otimes \ldots \otimes {\mathbb{C}}^2$ ($n$ times), that is a
$2^n$-dimensional complex space. A special basis,   called the
$2^n$-{\it computational basis}, is chosen for ${\otimes^n}
{\mathbb{C}}^2$. More precisely, it consists of the $2^n$ orthogonal
states $\vert \iota \rangle$, $0 \leq \iota \leq 2^n$ where $\iota$
is in binary representation and $\vert \iota \rangle$ can be seen as
tensor product of states (Kronecker product) $\vert \iota \rangle =
\vert \iota_1 \rangle \otimes \vert \iota_2 \rangle \otimes \ldots
\otimes \vert \iota_n \rangle$, whit $\iota_j \in \{0,1\}$. Then, a pure
state $\vert \psi \rangle \in {\otimes^n}{\mathbb{C}}^2$ is  a
superposition of the basis vectors $\vert \psi \rangle = \sum_{\iota
= 1}^{2^n} c_{\iota}\vert \iota \rangle$, with $\sum_{\iota =
1}^{2^n} \vert c_{\iota} \vert^2 = 1$.

In the usual representation of quantum
computational processes, a quantum circuit is identified with an
appropriate composition of {\it quantum gates}, mathematically
represented by \textit{unitary operators} acting on pure states of a
convenient ($n$-fold tensor product) Hilbert space ${\otimes^n}
{\mathbb{C}}^2$ \cite{NIC}. In other words, the standard model for
quantum computation is mathematically based on
``\textit{qubits-unitary operators}''.

As we said in Section \ref{HOL}, in general, a quantum system is not in a pure state. Moreover, there are interesting processes that
cannot be encoded by unitary evolutions. For example, the measurement at the end of the computation is a non-unitary operation, 
and the final state becomes a probability distribution over pure states i.e., a mixed state. 

In view of these facts,
several authors \cite{AKN, DF, FSA, GUD1, TA} have paid attention to a more general model of quantum computational processes, where pure
states are replaced by mixed states. In what follows we give a short description of this powerful model for quantum computers based on mixed states, which is better suited to our development.

As a particular case, we may associate to each vector of
the logical basis of ${\mathbb{C}}^2$ two density operators $P_0 = 
\vert 0 \rangle \langle 0 \vert$ and $P_1 = \vert 1 \rangle \langle 1
\vert$ that represent, in this framework, the falsity-property and the truth-property,
respectively.  Let us consider the operator $P_1^{(n)} =
\otimes^{n-1} I \otimes P_1$ on ${\otimes^n}{\mathbb{C}}^2$. 
By
applying the Born rule, we shall consider the probability of a density
operator $\rho$ as follows:
\begin{equation} \label{PROBDEF}
  p(\rho) = Tr(P_1^{(n)} \rho).
\end{equation}
Note that, in the particular case in which $\rho = \vert \psi \rangle
\langle \psi \vert$, where $\vert \psi \rangle = c_0\vert 0 \rangle +
c_1 \vert 1 \rangle$, we obtain that $\mathtt{p}(\rho)= \vert c_1 \vert
^2$. Thus, this probability value associated to $\rho$ is the
generalization of the probability value considered for
qubits.

A \textit{quantum operation} \cite{K} is a linear operator ${\cal
  E}:{\cal L}(H_1)\rightarrow {\cal L}(H_2)$ where ${\cal L}(H_i)$ is
the space of linear operators in the complex Hilbert space $H_i$ ($i=
1, 2$), representable as ${\cal E}(\rho)=\sum_{i}A_{i}\rho
A_{i}^{\dagger }$ where $A_i$ are operators satisfying
$\sum_{i}A_{i}^{\dagger }A_{i}=I$ (Kraus representation). It can be
seen that a quantum operation maps density operators into density
operators. Each unitary operator $U$ has a natural correspondent quantum operation ${\cal O}_{\cal U}$ such that, for each density operator $\rho$, ${\cal O}_{\cal U}(\rho) =
{\cal U}\rho {\cal U}^{\dagger}$ . In this way, quantum operations are generalizations of unitary operators.
It provides a powerfull model for quantum computation in which irreversible  porcesses can be also considered. This model based on density operators and quantum
operations  is known as ``\textit{quantum computation with mixed
states}'' (\cite{AKN, TA}).

\section{Quantum computational logics}\label{QcL}

An holistic extensions for classical logic in quantum computing, announced as the main goal of this paper, is fully supported in the formalism of quantum computation with mixed states. This naturally suggests a kind of quantum logical system related to quantum computation that allows us to achieve the holistic extension mentioned above. As expected, this logical system will be substantially different than standard Birkhoff-von Neumann quantum logic \cite{BvN}. In this section we show the differences between these two logics.

According to von Neumann's axiomatization, quantum events are mathematically realized by projectors of a Hilbert space ${\cal H}$. Hence, any experimental proposition concerning a quantum system corresponds to a projector in a convenient Hilbert space.
Closed subspaces of ${\cal H}$ are in one-to-one correspondence with the class of all projectors of $\mathcal{H}$
and they form an algebra called {\it Hilbert lattice} (denoted by $L(\mathcal{H}))$. In any Hilbert lattice the meet operation $\wedge$ corresponds to the set theoretical intersection between subspaces and the join operation $\vee$ corresponds to the smallest closed subspace of $\mathcal{H}$ containing the set theoretical union of subspaces. The ordering relation associated to the lattice $\mathcal{L}(\mathcal{H})$ is the inclusion of subspaces. Note that $L(\mathcal{H})$ is a bounded lattice where $\mathcal{H}$ is the \emph{maximum}, denoted by $1$, while $0$ denotes the \emph{minimum}, i.e., the subspace containing only the origin.
This lattice, equipped with the relation of orthogonal complement $^\bot$, can be described as an {\it ortholattice} \cite{KAL}. Then, the propositional structure that defines the standard quantum logic proposed  by Birkhoff and von Neumann, is given by 
the ortholattice $\langle {L(\mathcal{H})}, \lor, \land, ^\bot, 1,0 \rangle$. 
Let us notice that, unlike classical logic, in this structure the distributive law fails.
 However $L(\mathcal{H})$ satisfies a kind of weak distributivity. In case of a finite-dimensional Hilbert space $\mathcal{H}$, the ortholattice $L(\mathcal{H})$ is modular, i.e. satisfies the following condition known as {\it modular law}: $x\leq y \Longrightarrow x\lor (y\land z) = y \land (x\lor z) $. In the case of an infinite-dimensional Hilbert space, the modular law is not satisfied. In 1937, K. Husimi \cite{husimi73} showed that a weaker law, the so called {\it orthomodular law} ($x\leq y \Longrightarrow x\lor (x^\bot \land y) = y$), is satisfied in the ortholattice $L(\mathcal{H})$.

Quantum computation motivates several types of quantum logics. A family of these logics, deeply investigated in \cite{DUNN,DMW, HAG, GFS}, deals with qubits, following the ideas of Birkhoff and von Neumann. More precisely, this family of logics examines the orthostructure of the Hilbert lattices  $L(\otimes^n\mathbb{C}^2)$ of the space of $n$-dimensional qubits. In this way, the equational theories of $L(\otimes^n\mathbb{C}^2)$ are investigated in relation with the dimension $n$.

Differently, other quantum computational logical systems arise from the combinational structure associated to a set of quantum gates. They are defined taking into account algebraic properties of quantum operations acting on density operators belonging to $\otimes^n\mathbb C^2$. In what follows, we introduce a brief description of a type of logic associated to quantum circuits that provides the framework for our holistic extension of the classical logic. A problem, usually treated in classical computation and more precisely in digital techniques, is the following: 

\begin{itemize}
\item[]
{\it if $T$ is a combinational circuit, we want to know whether a given input state of $T$, represented by a string of bits $0$ and $1$, forces a determinate output state of $T$ given by a bit, that could be either $0$ or $1$}. 
\end{itemize}

As a general rule, this problem can be solved through effective procedures based on classical logic. Then, one may naturally extend this problem by considering circuits made from assemblies of fixed set of quantum gates. In this way, the input and the output of quantum circuits are labeled by density operators and possible notions of logical consequence are defined by relations between the input and the output of the circuits. Several families of quantum computational logics arise from these extensions \cite{DF, GUD1, LS}. Each of these logics are related to a fixed set of quantum gates and  
they have a common semantic based on probability-values, as introduced in Eq.(\ref{PROBDEF}). More precisely, a language for a quantum computational logic is a propositional language ${\cal L}_{\mathfrak{F}}(X)$ where $X$ is a non-empty set of variables and $\mathfrak{F}$ is a set of  connectives. Propositional variables are interpreted in a set ${\cal D}$ of density operators and for each connective $f\in \mathfrak{F}$, $f$ is naturally interpreted as a quantum operation $U_f$ closed on ${\cal D}$. An interpretation of ${\cal L}_{\mathfrak{F}}(X)$ in ${\cal D}$ is any function $e: {\cal L}_{\mathfrak{F}}(X) \rightarrow {\cal D}$ such that,  for each $f\in \mathfrak{F}$ with arity $k$, $e(f(x_1, \ldots , x_k)) = U_f(e(x_1), \ldots , e(x_k))$. To define a relation of semantical consequence $\models$ based on the probability assignment, it is necessary to introduce the notion of valuations. In fact, {\it valuations} are functions over the unitary real interval $v:{\cal L}_{\mathfrak{F}}(X) \rightarrow [0,1]$ such that $f$ can be factorized in the following way:

\begin{center}
\unitlength=1mm
\begin{picture}(20,20)(0,0)
\put(8,16){\vector(3,0){5}}
\put(2,10){\vector(0,-2){5}}
\put(10,4){\vector(1,1){7}}

\put(2,10){\makebox(13,0){$\equiv$}}

\put(2,16){\makebox(0,0){${\cal L}_{\mathfrak{F}}(X)$}}
\put(20,16){\makebox(0,0){$[0,1]$}}
 \put(2,0){\makebox(0,0){${\cal D}$}} \put(2,20){\makebox(17,0){$v$}} \put(2,8){\makebox(-6,0){$e$}}
\put(18,2){\makebox(-4,3){$p$}}
\end{picture}
\begin{equation}\label{QCLM}
\end{equation}
\end{center}

Since an interpretation always determines a valuation, for each interpretation $e$, we denote by $e_p$ the valuation associated to $e$. The abstract notion of semantical consequence $\models$ related to ${\cal D}$ is given by: $$\alpha \models \varphi \hspace{0.3cm}  iff \hspace{0.3cm} {\cal R}[v(\alpha), v(\varphi)]$$ where ${\cal R} \subseteq [0,1]^2$ is a reflexive and transitive relation. Note that, the natural extension of classical logical consequence can be formulated as follows: 

\begin{equation}\label{LOGCONSEQUENCE}
\alpha \models \varphi \hspace{0.3cm}  iff \hspace{0.3cm} e_p(\alpha) = 1 \Longrightarrow e_p(\varphi) = 1. 
\end{equation}

These quantum logical systems are known as {\it Quantum computational logic} \cite{DGG,GUD1,DF}. 

We can establish a comparison between the Birkhoff-von Neumann quantum logic and the quantum computational logic associated to Hilbert spaces of the form $\otimes^n\mathbb{C}^2$. Basically, Birkhoff-von Neumann quantum logic interprets propositions as closed subspaces of $\otimes^n\mathbb{C}^2$ and connectives as operations of its natural orthostructures. On the other hand, quantum computational logic with mixed states interprets propositions as $n$-qbits (where each $n$-qbit is related to one dimensional sub space of $\otimes^n\mathbb{C}^2$) and connectives as quantum operations acting on $\otimes^n\mathbb{C}^2$.

\section{A fuzzy extension for classical logic in\\ quantum computation}\label{FE}

Quantum computational logic systems can be framed as generalizations of probabilistic logics. {\it Probabilistic logics} is the name that Adams \cite{AD} proposed for the formal investigation on the transmission of probability values thorough valid inferences. This idea can be generalized by considering non-Kolmogorovian probability models as it happens in the case of quantum computational logics, whose semantic is based on the Born rule. Thus, if ${\cal L}_{\mathfrak{F}}(X)$ is a language associated to quantum computational logic system, then the probabilistic semantic for ${\cal L}_{\mathfrak{F}}(X)$ assumes its truth value in the continuous $[0,1]$. It also suggests a strong relation between quantum computational logic and fuzzy logical systems. In this section a fuzzy extension for classical propositional logic coming from quantum computation is introduced. It provides the underling formalism for the holistic extension for the classical propositional logic in quantum computation developed in Section \ref{HTE}.

In a general case, for a quantum computational logical system that extended the classical logic, it is quite natural to require the following condition: 
\begin{itemize}
\item[]
{\it once fixed a language ${\cal L}_{\mathfrak{F}}(X)$, the elements of the set  $\mathfrak{F}$ have to be interpreted as quantum operations that are able to fully describe, from the  truth-functionally point of view, classical logic.}
\end{itemize}

In other words, the set of connectives $\mathfrak{F}$, restricted to the classical truth values $\{0,1\}$, is functionally complete\footnote{We say that a set of classical connectives is {\it functionally complete} if it is sufficient to express every truth-function.} with respect to propositional classical logic. Functional completeness, besides being an important logical property, turns out to be crucial also for technological applications. A paradigmatic case is represented by the digital techniques where logical gates can be represented by propositional connectives and  circuits by propositional formulas. For technical reasons (standardization of integrated circuits, energy optimization) sometimes it is necessary to built circuits by using a restricted set of logical gates.
We focus our attention on the set $\langle \neg, \land  \rangle$ which is functionally complete  for classical logic. Thus, by induction, a logical system $\langle \neg, \land  \rangle$ can represent all truth-functions of classical logic. However, the set $\langle \neg, \land  \rangle$ could not be functionally complete for some extension of classical logic. The rest of this section is devoted to investigate a natural extension of $\langle \neg, \land  \rangle$ to quantum  computational logic with mixed states. The mentioned extension will be built from
a logical system equipped with only one connective, semantically interpreted as the well known Toffoli quantum gate.

\begin{definition}
{\rm For each density operator $\rho$ in $\otimes^m{\mathbb{C}^2}$ the negation $\mathbb {NOT}^{(2^m)}(\rho)$ is define as follows:
$$\mathbb {NOT}^{(2^m)}(\rho) =  (I^{(2^{m-1})}\otimes NOT) \hspace{0.1cm} \rho \hspace{0.1cm}(I^{(2^{m-1})}\otimes NOT)$$ where 
$NOT = \left[\begin{array}{cc}
    0 & 1 \\
    1 & 0
  \end{array}\right].$
}
\end{definition}

In \cite{DGG}  it is proved that 
\begin{equation} \label{PROBNEG}
p(\mathbb {NOT}^{(2^m)}(\rho)) = 1 - p(\rho).
\end{equation}

An extension of the classical conjunction can be implemented via {\it Toffoli gate}.  It was introduced by Tommaso Toffoli \cite{TOF} and it is represented by the ternary classical connective $T(x,y,z) = (x,y, xy \widehat{+} z) $ where $\widehat{+}$ is the sum modulo $2$. When $z=0$, $T(x,y,0)$ reproduces the classical conjunction. Toffoli gate is natural extended to qubits in the following way.

For any natural numbers $m,k \geq 1$ and for any vectors of the standard orthonormal basis $\ket{x} = \ket{x_1\ldots x_m } \in \otimes^m{\mathbb{C}^2} $
, $\ket{y} = \ket{y_1 \ldots y_k} \in \otimes^k{\mathbb{C}^2} $ and $\ket{z} \in {\mathbb{C}^2} $ , the Toffoli gate $T^{(m,k,1)}$ on $\otimes^{m + k + 1}{\mathbb{C}^2}$ is defined as follows: $$T^{(m,k,1)}(\ket{x} \otimes \ket{y}\otimes \ket{z}) = \ket{x} \otimes \ket{y}\otimes \ket{x_m y_k \widehat{+} z}.$$   

By {\rm \cite[Proposition 3.1]{FS}},
for any natural number $m,k \geq 1$, $T^{(m,k,1)}$ is a unitary operator whose matrix representation  is given by

\begin{eqnarray} \label{MATTOF}
T^{(m,k,1)} & = & I^{(2^{m+k + 1})} +  P_1^{(2^m)} \otimes P_1^{(2^k)} \otimes (Not - I)  \\
& = &  I^{(2^{m-1})} \otimes \left[
\begin{array}{c|c}
I^{(2^{k+1})} & {\bf 0}  \\ \hline
{\bf 0}  & I^{(2^{k-1})} \otimes Xor   
\end{array}\right]
\end{eqnarray}

where $Xor = \left[
\begin{array}{cccc}
1 & 0 & 0 &0  \\ 
0 & 1 & 0 &0  \\
0 & 0 & 0 &1  \\     
0 & 0 & 1 &0
\end{array}\right]
$.

\bigskip

$T^{(m,k,1)}$ allows us to extend the classical conjunction as follows.

\begin{definition}
{\rm Let $\rho^m$ be a density operator in $\otimes^m\mathbb C^2$ and $\rho^k$ be a density operator in $\otimes^k\mathbb C^2.$
We define: $$\mathbb {AND}^{(m,k)}(\rho^m\otimes\rho^k) = T^{(m,k,1)}(\rho^m\otimes\rho^k\otimes P_0)T^{(m,k,1)}.$$
}
\end{definition}

In \cite{DGG} it is proved that  
\begin{equation}\label{PROBAND}
p(\mathbb {AND}^{(m,k)}(\rho^m\otimes\rho^k))=p(\rho^m)p(\rho^k).
\end{equation}

Let us consider the set ${\cal D}_n$ of all density operators on $\otimes^n{\mathbb{C}^2}$. It is very important to remark that $\mathbb {AND}^{(m,k)}$ can be seen as a binary operator of the form 

\begin{equation}\label{2ary}
\mathbb {AND}^{(m,k)}:{\cal D}_m \times {\cal D}_k \rightarrow {\cal D}_{m+k+1}.
\end{equation}

In order to define a quantum computational logical system in the sense of Section \ref{QLMS} and based on $\{ \mathbb {AND}^{(-,-)}, \mathbb {NOT}^{(2^-)} \}$, we consider the set ${\cal D} = \bigcup_n {\cal D}_n$ and we introduce the binary connective $\mathbb {AND}$ and the unary connective $\mathbb {NOT}$ in ${\cal D}$ as
\begin{enumerate}
\item[] $\mathbb {AND}(\rho, \sigma) = \mathbb {AND}^{(m,k)}(\rho \otimes\sigma)$ iff, $\rho \in {\cal D}_m$ and $\sigma \in {\cal D}_k$

\item[] $\mathbb {NOT}(\rho) = \mathbb {NOT}^{(2^m)}(\rho)$ iff $\rho \in {\cal D}_m$.
\end{enumerate}

Note that $\mathbb {AND}$ and $\mathbb {NOT}$ are closed operations in ${\cal D}$. Thus, these operations define a quantum computational logical system in the sense of Section \ref{QLMS} that we shall denote as ${\cal QC}_{AN}$. By Eq (\ref{PROBNEG}) and Eq (\ref{PROBAND}) it is immediate to see that

\begin{equation}\label{COMPSIST}
{p(\mathbb {NOT}(\rho)) = 1-p(\rho)}, \hspace{1cm} { p(\mathbb {AND}(\rho, \sigma)) = p(\rho) p(\sigma).}
\end{equation}

From a probabilistic point of view,  $\mathbb {NOT}^{(2^m)}$ gate can be described as an instance of Toffoli gate. In fact, by \cite[Theorem 3.1]{FS}, for each density operator $\rho$ in $\otimes^m\mathbb C^2$ we can easily see that $$p(\mathbb {NOT}^{(2^m)}(\rho)) = p(T^{(m,k,1)}(\rho\otimes P_1^{(k)}\otimes P_1)T^{(m,k,1)}).$$ Thus $\mathbb {AND}$ and $\mathbb {NOT}$ can be considered as two particular instance of Toffoli gate. Consequently, ${\cal QC}_{AN}$ can be sees as a logic construction arising from Toffoli gate only.

In the particular case in which $p(\rho)$ and $p(\sigma)$ are $1$ or $0$, these quantum gates behave as the classical negation and conjunction, respectively. In this way, ${\cal QC}_{AN}$ provides an extension of the classical propositional logic. 

It is possible to characterize the subset of ${\cal D}$ for which the set of connectives $\{\mathbb {NOT}, \mathbb {AND}\} $ behaves classically. In fact: let $\rho \in {\cal D}_n$ and suppose that the diagonal of $\rho$ is given by $diag(\rho) = \{r_{1,1}, r_{2,2} \ldots r_{2^n,2^n}\}$. Note that $p(\rho) \in \{0,1\}$ iff $\sum_{i=1}^{2^{n-1}} r_{2i, 2i} \in \{0,1\}$. If we define the set $${\cal D}_n^{class} = \{ \rho = (r_{i,j})_{1\leq i,j \leq 2^n} \in {\cal D}_n : \sum_{i=1}^{2^{n-1}} r_{2i, 2i}  \in \{0,1\} \}$$ then      
\begin{equation} \label{classic}
{\cal D}^{class} = \bigcup_{n} {\cal D}_n^{class}
\end{equation} \label{classic} is the subset of ${\cal D}$ in which $\{\mathbb {NOT}, \mathbb {AND}\} $ behaves classically. \\

${\cal QC}_{AN}$ is strongly related to the {\it Basic fuzzy logic} introduced by P. H\'ajek at the end of the 1990s \cite{HAJ}. This kind of fuzzy logic is conceived as a theory of the approximate reasoning based on many-valued logic systems. Basic fuzzy logic is the logic associated to {\it continuous $t$-norms} i.e., continuous, commutative, associative, and monotone binary operations on $[0,1]$ with 1 as the neutral element. These operations are taken as possible truth-functions of conjunctions in these systems. Each continuous $t$-norm  determines a semantic of fuzzy propositional logic. For example the {\it \L ukasiewcz $t$-norm} $x \odot_{\L} y = \max\{0, x+y -1\}$ defines the conjunction of the \L ukasiewcz infinite many valued logic, where $\neg_{\L}x = 1-x$ is the negation in this logic. The {\it product $t$-norm} $x \odot_p y = x y$ defines the conjunction of the Product logic \cite{CT} and the G\"{o}del $t$-norm $x \odot_G y = \min\{x,y\}$ defines the conjunction of the linear Heyting logic. For the sake of simplicity, in the calculations, we use the common product $x y$ to indicate the product $t$-norm $x \odot_p y$.   

Since $p(\mathbb {NOT}(\rho)) = 1-p(\rho)$, we can identify $\mathbb {NOT}$ with the \L ukasiewicz negation and since $p(\mathbb {AND}(\rho, \sigma)) = p(\rho) p(\sigma)$, $\mathbb {AND}$ can be identified with the product $t$-norm. Thus, from a strictly semantic point of view, we can establish the following identification: 
\begin{equation}\label{IDENTIFICATION}
\{\mathbb {NOT}, \mathbb {AND}\} \stackrel{\approx}{\mbox{ {{\footnotesize semantic}}}} \{\neg_{\L}, \odot_p\}.
\end{equation}

We remark that connectives $\{\neg_{\L}, \odot_p\}$ define a multiplicative fragment of the fuzzy logical system known as {\it product many valued logic or $PMV$-logic}, studied in  \cite{DD, MONT}.

This semantic connection between two logical systems is even deeper and it is formally rooted in the equivalence relation on ${\cal D}$ given by   
\begin{equation}\label{EQUIVD}
\rho \approx \sigma \hspace{0.5cm} \mbox{iff} \hspace{0.5cm} p(\rho) = p(\sigma). 
\end{equation}
It is not very hard to see that, the quotient set ${\cal D}/_{\approx}$ can be identified to the real interval $[0,1]$ and $\approx$ is a congruence with respect to $\{\mathbb {NOT}, \mathbb {AND}\}$. Thus, both operations naturally induce two operations over the equivalence classes in ${\cal D}/_{\approx}$ given by ${\mathbb {NOT}}_{\approx}([\rho]) = [{\mathbb {NOT}}(\rho)]$ and ${\mathbb {AND}}_{\approx}([\rho], [\sigma]) = [{\mathbb {AND}}(\rho,\sigma)]$.  Then, the algebraic structures $\langle {\cal D}_{\approx}, \mathbb {NOT}_{\approx}, \mathbb {AND}_{\approx}\rangle$ and $\langle [0,1], \neg_{\L}, \odot_p \rangle$ coincide and they induce the same algebraic semantic for both logical systems. As a consequence, the natural $\{\mathbb {NOT}, \mathbb {AND}\}$-homomorphism $\pi: {\cal D} \rightarrow {\cal D}/_{\approx} = [0,1]$ is identifiable with the assignment of probability in ${\cal QC}_{AN}$.  In this way, ${\cal QC}_{AN}$ is semantically related to Basic fuzzy logic providing a fuzzy extension for the propositional classical logic in quantum computation with mixed states. \\

In classical logic, concepts of contradiction and tautology can be syntactically represented in terms of $\{ \neg, \land \}$. Contradictions are those formulas equivalent to $p \land \neg p$ and tautologies are those formulas equivalent to $\neg (p \land \neg p)$. From these facts, the formula
$p \land \neg p$ is sometimes refereed as {\it syntactic contradiction} and $\neg (p \land \neg p)$ (more precisely the equivalent form $p\lor \neg p$) is refereed as {\it syntactic tautology}. In this work we accord with this terminology.

In ${\cal QC}_{AN}$, a syntactic representation for contradictions and tautologies is lost. This fact can be explained taking into account that real numbers do not contain zero divisors. Then, there is not an algebraic expression built from $\{\neg_{\L}, \odot_p\}$ that produces the constant functions $1$ or $0$. Hence, by the semantic identification given in Eq (\ref{IDENTIFICATION}),  there does not exist a formula in the language of $\{\mathbb {AND}, \mathbb{NOT}\}$ that produces a contradiction or a tautology in ${\cal QC}_{AN}$. However, the ${\cal QC}_{AN}$-extensions of the syntactic contradiction and the syntactic tautology, have interesting properties. The ${\cal QC}_{AN}$ are:

\begin{itemize}

\item[]
$p \land \neg p \hspace{0.8cm} \stackrel{{{\cal QC}_{AN}}}{\longrightarrow} \hspace{0.3cm} {\mathbb {AND}}(\rho, {\mathbb {NOT}}\rho)$ \hspace{1.25cm} {\footnotesize [syntactic contradiction]},

\item[]
$\neg (p \land \neg p) \hspace{0.3cm} \stackrel{{{\cal QC}_{AN}}}{\longrightarrow} \hspace{0.3cm}  {\mathbb {NOT}}({\mathbb {AND}}(\rho, {\mathbb {NOT}}\rho))$ \hspace{0.2cm} {\footnotesize [syntactic tautology]}.

\end{itemize}

Since ${\mathbb {NOT}}$ is an involution on ${\cal D}$, the ${\cal QC}_{AN}$-extension of the syntactic contradiction and ${\cal QC}_{AN}$-extension of the syntactic tautology are dual concepts. Thus, for sake of simplicity, we can focus our attention on the notion of contradiction only. 
By Eq(\ref{IDENTIFICATION}) we can see that:  
\begin{equation}\label{CONTRADICTION}
p({\mathbb {AND}}(\rho, {\mathbb {NOT}}\rho)) = p(\rho)(1-p(\rho)) \leq \frac{1}{4}.
\end{equation}
Thus, $p({\mathbb {AND}}(\rho, {\mathbb {NOT}}\rho)) = 0$ iff $p(\rho) \in \{0,1\}$. In other words, the fuzzy extension of the classic syntactic contraddiction ${\mathbb {AND}}(-, {\mathbb {NOT}}(-))$ has a classical behaviour over the set ${\cal D}^{class}$ only.

\section{An holistic type extension for classical logic}\label{HTE}

Quantum computational logic with mixed states can also provide an interesting holistic type extension for the classical propositional logic. This extension arises when non factorizable states are considered as inputs in the Toffoli quantum gate. We will also note that, the fuzzy system $\{\neg_{\L}, \odot_P \}$ plays an important role for describing the mentioned holistic extension.

The formal language in which classical logic and most logical systems are expressed, is regulated by strict syntax rules. 
The basic idea at the origin of these languages is the fact that each proposition or formula can be built by means of a recursive procedure from a distinguished set of propositions, called atomic propositions. In this way, complex propositions are recursively obtained  from atomic propositions assembled by connectives. For each connective a natural number, the {\it arity}, is assigned. The arity defines the number of propositions that the connectives assemble. When an algebraic semantic for these logical systems is considered, an  $n$-ary connective is interpreted as an algebraic operation having $n$ arguments. Thus, the arity is an invariant property associated to a connective. All these ideas was already taken into accout in ${\cal QC}_{AN}$, where separability conditions of the states were considered. More precisely, $\mathbb {AND}^{(m,k)}$ is viewed as a $2$-ary connective in the ideal case in which a factorizable state of the form $\rho_m \otimes \rho_k$ is considered as input.  

In general, of course, this is not the case. Quantum systems continually interact with the  environment, building up correlations. For a more realistic approach, we can assume that the input of the $\mathbb {AND}^{(m,k)}$ can be also a non-factorizable mixed state $\rho$ in  $\otimes^{m+k}\mathbb C^2$ taking into account its holistic type representation given in Eq.(\ref{decomposition}) i.e. $$\rho = \rho_m \otimes \rho_k + {\bf M}(\rho)$$ where $\rho_m$ and $\rho_k$ are the reduced states of $\rho$ in $\otimes^{m}\mathbb C^2$ and $\otimes^{k}\mathbb C^2$ respectively. 

Differently, with respect to Eq.(\ref{2ary}), when non factorized states are taken into account, $\mathbb {AND}^{(m,k)}$ behaves as a unary operator of the form $ \mathbb {AND}^{(m,k)}: {\cal D}_{m+k} \rightarrow {\cal D}_{m+k+1}$. This behavior of $\mathbb {AND}^{(m,k)}$ motivates an holistic type extension of classical conjunction.  The following definition formally introduce an operator that describes the unary behavior of $\mathbb {AND}^{(m,k)}$.

\begin{definition}\label{HOLAND}
 For any density operator $\rho\in\otimes^{m+k}\mathbb C^2$ we define: $${\mathbb {AND}^{(m,k)}_{Hol}}(\rho)=T^{(m,k,1)}(\rho\otimes P_0)T^{(m,k,1)}.$$
\end{definition}

For sake of simplicity, we use the following notation: if $\rho$ is a density operator in $\otimes^{m+k}\mathbb C^2$ then ${\mathbb {T}}_p^{(m,k,1)}(\rho)$
denotes the matrix $${\mathbb {T}}_p^{(m,k,1)}(\rho) = P_1^{2^{m+k+1}}(T^{(m,k,1)}({\bf M}(\rho) \otimes P_0)T^{(m,k,1)}). $$ Then, by Eq.(\ref{decomposition}) and Eq.(\ref{PROBAND}) follows that if $\rho$ is a density operator in $\otimes^{m+k}\mathbb C^2$ and $\rho_m$, $\rho_k$ are the reduced states of $\rho$ in $\otimes^{m}\mathbb C^2$ and $\otimes^{k}\mathbb C^2$, respectively, then:   

\begin{equation}\label{HOLFUZZY11}
{\mathbb {AND}^{(m,k)}_{Hol}}(\rho) = {\mathbb {AND}^{(m,k)}}(\rho_m \otimes \rho_k) + T^{(m,k,1)}({\bf M}(\rho) \otimes P_0) T^{(m,k,1)}
\end{equation}

and the probability of the holistic conjunction will assume the form:

\begin{equation}\label{HOLFUZZY13}
p({\mathbb {AND}^{(m,k)}_{Hol}}(\rho)) = p(\rho_m) p(\rho_k) + tr({\mathbb {T}}_p^{(m,k,1)}(\rho)).
\end{equation}

Further, in the special case where $\rho=\rho_m\otimes\rho_k $, Eq.(\ref{HOLFUZZY11}) clearly collapse in:
\begin{equation}\label{HOLFUZZY12}
 {\mathbb {AND}^{(m,k)}_{Hol}}(\rho) = {\mathbb {AND}^{(m,k)}}(\rho_m\otimes\rho_k) .
\end{equation}

The above result shows that ${\mathbb {AND}^{(m,k)}}$ is implicitly acting in ${\mathbb {AND}^{(m,k)}_{Hol}}(\rho)$ over the reduced states of $\rho$.\\

In what follows we provide a simple way to estimate  $p({\mathbb {AND}^{(m,k)}_{Hol}}(\rho))$,  $p(\rho_m)$, $p(\rho_k)$ and  $tr({\mathbb {T}}_p^{(m,k,1)}({\bf M}(\rho)))$. We first introduce the following technical definition.  

\begin{definition}
{\rm Let $\rho = (r_{i,j})_{1\leq i,j \leq 2^{m+k}}$ be a density operator in $\otimes^{m+k}\mathbb C^2$ divided in $2^m\times 2^m$ blocks $T_{i,j}$ where each of them is a $2^k$-square matrix. 
\begin{equation*}
    \rho = 
    \left(\begin{array}{cccc}
       T_{1,1} & T_{1,2} & \ldots & T_{1,2^m}  \\
       T_{2,1} &  T_{2,2} & \ldots &  T_{2,2^m}  \\
      \vdots & \vdots & \vdots & \vdots  \\
       T_{2^m,1} &  T_{2^m,2} & \ldots &  T_{2^m,2^m}   \\
          \end{array}\right).
\end{equation*}

Then, the {\it $(m,k)$-Toffoli blocks of $\rho$} are the diagonal blocks $(T_i =T_{i,i})_{1\leq i \leq 2^m}$ of $\rho$. Moreover, we introduce the following parameters:

\begin{enumerate}
\item[]
$\beta^{m,k}(\rho)=\sum_{j=1}^{2^m-1} \sum_{i=0}^{2^{k-1}-1}r_{(2i+1)+j2^k}$ i.e. the sum of the odd diagonal elements of the even
$(m,k)$-Toffoli blocks $T_{2i}$ of $\rho$,

\item[]
$\gamma^{m,k}(\rho)=\sum_{j=0}^{2^m-2}\sum_{i=1}^{2^{k-1}}r_{2i+j2^k}$ the sum of the even diagonal elements of the odd
$(m,k)$-Toffoli blocks $T_{2i+1}$ of $\rho$,

\item[]
$\delta^{m,k}(\rho)=\sum_{j=1}^{2^m-1}\sum_{i=1}^{2^{k-1}}r_{2i+j2^k}$ the sum of the odd diagonal elements of the odd
$(m,k)$-Toffoli blocks $T_{2i+1}$ of $\rho$.

\end{enumerate}
}
\end{definition}

By {\rm \cite[Proposition 4.3]{FS}}
if we consider a density operator $\rho$ in $\otimes^{m+k}\mathbb C^2$ with $m,k \geq 1$ and if $r_i$ is the $i$-th diagonal element of $\rho$, then:

\begin{equation}\label{TOFP0}
p({\mathbb {AND}^{(m,k)}_{Hol}}(\rho))= \sum_{j=1}^{2^{m-1}} \sum_{i=1}^{2^{k-1}}r_{(2j-1)2^k+2i}.
\end{equation} 
More precisely, $p({\mathbb {AND}^{(m,k)}_{Hol}}(\rho))$ is the sum of the even diagonal elements of the even $(m,k)$-Toffoli blocks $T_{2i}$ of $\rho$.

Eq.(\ref{TOFP0}) is a useful tool that allows us to evaluate in very simple way all the terms involved in Eq.(\ref{HOLFUZZY13}), as the next theorem provides\footnote{For more technical details, see {\rm \cite[Proposition 4.4]{FS}}.}.

\begin{theo}\label{error}
Let $\rho$ be a density operator in $\otimes^{m+k}\mathbb C^2$. Let $\rho_m$ be the reduced state of $\rho$ on $\otimes^{m}\mathbb C^2$ and let $\rho_k$ be the reduced state of $\rho$ on $\otimes^{k}\mathbb C^2$. Then,

\begin{enumerate}
\item
$1 = p({\mathbb {AND}^{(m,k)}_{Hol}}(\rho)) +  \beta^{m,k}(\rho) + \gamma^{m,k}(\rho) + \delta^{m,k}(\rho)$

\item
$p(\rho_m) = p({\mathbb {AND}^{(m,k)}_{Hol}}(\rho)) +  \beta^{m,k}(\rho)$,

\item
$p(\rho_k) = p({\mathbb {AND}^{(m,k)}_{Hol}}(\rho)) + \gamma^{m,k}(\rho)$,

\item
$ tr({\mathbb{T}}_p^{(m,k,1)}(\rho)) = p({\mathbb {AND}^{(m,k)}_{Hol}}(\rho))\delta^{m,k}(\rho) - \beta^{m,k}(\rho)\gamma^{m,k}(\rho)$.

\end{enumerate}

\qed
\end{theo}
Interestingly enough, Theorem \ref{error} allows us to obtain some boundary estimation on the quantities involved in Eq.(\ref{HOLFUZZY13}). 

By Theorem \ref{error} (2-3) is immediate to see that
\begin{equation}
p(AND_{Hol}^{(m,n)} (\rho)) \leq \mathtt{p}(\rho_m), \mathtt{p}(\rho_k).
\end{equation}

Further, the incidence of the holistic component ${\bf M}(\rho)$ on the probability of $p(AND_{Hol}^{(m,n)} (\rho))$  lives in the bounded interval:

\begin{equation}
-\frac{1}{4}\leq  tr({\mathbb{T}}_p^{(m,k,1)}(\rho)) \leq \frac{1}{4}.
\end{equation}
To show this, we have to consider the following maximum/minimum problem  $$ \begin{cases}tr({\mathbb{T}}_p^{(m,k,1)}(\rho))= \delta^{m,k}(\rho)^2 - \beta^{m,k}(\rho)  \gamma^{m,k}(\rho)&  \\
2\delta^{m,k}(\rho) + \beta^{m,k}(\rho)+  \gamma^{m,k}(\rho) = 1 &\end{cases}.$$ Note that $\max tr({\mathbb{T}}_p^{(m,k,1)}(\rho))$ is given when $\beta^{m,k}(\rho)+  \gamma^{m,k}(\rho) = 0$. Thus $\max \{tr({\mathbb{T}}_p^{(m,k,1)}(\rho))\} = \delta^{m,k}(\rho)^2 = \frac{1}{4}$.  While $\min \{tr({\mathbb{T}}_p^{(m,k,1)}(\rho))\}$ is given under the condition $ \beta^{m,k}(\rho)+  \gamma^{m,k}(\rho) = 1$. In this way $\min \{tr({\mathbb{T}}_p^{(m,k,1)}(\rho))\} - \max\{\beta^{m,k}(\rho)(1-\beta^{m,k}(\rho)) \} = -\frac{1}{4}$. 

Finally, 
in the special case where $p(AND_{Hol}^{(m,k)}(\rho))=1$ the holistic component of $\rho$ has not any probability incidence, i.e.   $ tr({\mathbb{T}}_p^{(m,k,1)}(\rho))=0$. In this case
$p(\rho^m)=p(\rho^k)=1$. In fact, suppose that $p(AND_{Hol}^{(m,k)}(\rho))=1$; then, by Theorem \ref{error}-(1 and 2), $p(\rho^m)=p(\rho^k)=1$ and $\beta^{m,k}(\rho) = \gamma^{m,k}(\rho) = 0$. Thus, $1= p({\mathbb {AND}^{(m,k)}_{Hol}}(\rho)) + \delta^{m,n}(\rho) + \beta^{m,k}(\rho) + \gamma^{m,k}(\rho) = 1 + \delta^{m,n}(\rho) +0+0$ and then $\delta^{m,n}(\rho) = 0$. Hence, by Theorem \ref{error}-3, $ tr({\mathbb{T}}_p^{(m,k,1)}(\rho))=0$.

\bigskip

To define an holistic extension of the classical conjunction starting from ${\mathbb {AND}}^{(m,k)}_{Hol}$, we have to deal with the following situation: if $\rho$ is a density operator on $\otimes^n\mathbb C^2$ where $n=m+k=m'+k'$ and $m\neq m', k\neq k'$ then, we generally have that $${\mathbb {AND}}^{(m,k)}_{Hol}(\rho)\neq {\mathbb {AND}}^{(m',k')}_{Hol}(\rho).$$ In other words, a logical connective based on ${\mathbb {AND}}^{(-,-)}_{Hol}$ also requires a precise information about the holistic representation of the argument in the sense of Eq.(\ref{decomposition}). For this, we introduce the following notions: $\rho_{\langle m,k \rangle}$ indicates that $\rho$ is a density operator in $\otimes^{m+k}\mathbb C^2$ where the holistic representation $\rho = \rho_m \otimes \rho_k + {\bf M}(\rho)$ is choosen. We also define the set ${\cal D}_{Hol}$ as: $${\cal D}_{Hol} = \{\rho_{\langle m,k \rangle}: m,k \in {\mathbb {N}} \}.$$ If we consider the relation in ${\cal D}_{Hol}$ given by 
\begin{equation}\label{EQUIVD}
\rho_{\langle m,k \rangle} \approx_H \rho_{\langle m',k' \rangle} \hspace{0.2cm} \mbox{iff} \hspace{0.2cm}  m+k = m'+k'
\end{equation}
then $\approx_H$ is an equivalence and ${\cal D}_{Hol}/_{\approx_H} = {\cal D}$.\\

We also note that, if $\rho$ is a density operator on $\otimes^{m+k}\mathbb C^2$, Proposition \ref{MATTOF} suggests a privileged (holistic) interpretation of the codomain for ${\mathbb {AND}}^{(m,k)}_{Hol}(\rho)$. In fact:
\begin{eqnarray*}
{\mathbb {AND}}^{(m,k)}_{Hol}(\rho) & = & T^{(m,k,1)} (\rho \otimes P_0) T^{(m,k,1)}\\
& = & (I^{(2^{m+k + 1})} +  P_1^{(2^m)} \otimes P_1^{(2^k)} \otimes (Not - I))(\rho \otimes P_0)(I^{(2^{m+k + 1})}+   \\
&  &  P_1^{(2^m)} \otimes P_1^{(2^k)} \otimes (Not - I))\\
& = & \rho \otimes P_0 + {\bf M}({\mathbb {AND}}^{(m,k)}_{Hol}(\rho))
\end{eqnarray*}
where ${\bf M}({\mathbb {AND}}^{(m,k)}_{Hol}(\rho)) = P_1^{(2^m)} \otimes P_1^{(2^k)} \otimes (Not - I))(\rho \otimes P_0)(P_1^{(2^m)} \otimes P_1^{(2^k)} \otimes (Not - I))$. This suggests to consider $({\mathbb {AND}}^{(m,k)}_{Hol}(\rho))_{\langle m+k,1 \rangle}$ as a natural holistic representation for ${\mathbb {AND}}^{(m,k)}_{Hol}(\rho)$. Thus, we define the holistic extension of the classical conjunction as follows:  

\begin{itemize}
\item[]${\mathbb {AND}}_{Hol}(\rho_{\langle m,k \rangle}) = ({\mathbb {AND}}^{(m,k)}_{Hol}(\rho))_{\langle m+k,1 \rangle}$.  

\end{itemize}

In this way ${\mathbb {AND}}_{Hol}$ defines a unary connective in ${\cal D}_{Hol}$.  Note that Eq.(\ref{HOLFUZZY11}) provides a deep relation between the connectives ${\mathbb {AND}}_{Hol}$ and ${\mathbb {AND}}$. In fact, for $\rho_{\langle m,k \rangle} = \rho_m \otimes \rho_k + {\bf M}(\rho)$ we have that   
\begin{eqnarray*}
{\mathbb {AND}}_{Hol}(\rho_{\langle m,k \rangle})& = & {\mathbb {AND}^{(m,k)}}(\rho_m \otimes \rho_k) + T^{(m,k,1)}({\bf M}(\rho) \otimes P_0) T^{(m,k,1)}\\
& = & {\mathbb {AND}}(\rho_m \otimes \rho_k) + T^{(m,k,1)}({\bf M}(\rho) \otimes P_0) T^{(m,k,1)}.
\end{eqnarray*}

The connective ${\mathbb {NOT}}$, formally defined on ${\cal D}$, has a natural extension to ${\cal D}_{Hol}$. Taking into account the equivalence $\approx_H$ in ${\cal D}_{Hol}$, introduced in Eq.\ref{EQUIVD}, for each $\rho_{\langle m,k \rangle} \in {\cal D}_{Hol}$ we can define ${\mathbb {NOT}}(\rho_{\langle m,k \rangle}) = {\mathbb {NOT}}([\rho_{\langle m,k \rangle}]_{\approx_H})$ where the equivalence class $[\rho_{\langle m,k \rangle}]_{\approx_H}$ is identified to a density operator on ${\cal D}$. In this way $\approx_H$ becomes a congruence with respect to ${\mathbb {NOT}}$, and ${\mathbb {NOT}}$ is well defined on ${\cal D}_{Hol}$.          

The pair ${\mathbb {AND}}_{Hol}$, ${\mathbb {NOT}}$ defines an holistic type extension for classical logic in the framework of quantum computation with mixed states. We denote this logical system as ${\cal QC}^{Hol}_{AN}$. We want to remark two peculiarities about the system ${\cal QC}^{Hol}_{AN}$.
First: while classical logic needs at least one binary connective to describe any possible truth-function, ${\cal QC}^{Hol}_{AN}$ can describes any possible classical truth-function by involving two unary connectives. Second: since ${\cal QC}^{Hol}_{AN}$ is described by unary connectives, the notion of classical syntactic contradiction - that had a natural extension in ${\cal QC}_{AN}$ - seems to have not an extension in ${\cal QC}^{Hol}_{AN}$. The rest of the section is devoted to this topic. \\

${\cal QC}^{Hol}_{AN}$ is a logical system having unary connectives only. This fact does not allow us to extend, in a natural way, the syntactic representation of the classical contradiction given by $p \land \neg p$. But it is possible to characterize a sub class of ${\cal D}_{Hol}$ that preserves the notion of syntactic contradiction when ${\mathbb {AND}}_{Hol}$ takes arguments  on this class. 

Let us remind that the syntactic contradiction, extended to ${\cal QC}_{AN}$, is given by ${\mathbb {AND}}(\rho, {\mathbb {NOT}}(\rho))$ where $p({\mathbb {AND}}(\rho, {\mathbb {NOT}}(\rho)))= p(\rho) (1-p(\rho))$. Following this idea, we want to characterize the elements $\rho_{\langle m,k \rangle}$  in ${\cal D}_{Hol}$ such that $p(\rho_m) = 1- p(\rho_k)$. In this way, if $\rho_{\langle m,k \rangle}$ is of the form $\rho_{\langle m,k \rangle} = \rho_m \otimes {\mathbb {NOT}}(\rho_m)$ then ${\mathbb {AND}}_{Hol}(\rho_{\langle m,k \rangle}) = {\mathbb {AND}} (\rho_m \otimes {\mathbb {NOT}}(\rho_m))$. It generalizes the fuzzy extension of the syntactic contradiction in ${\cal QC}^{Hol}_{AN}$. We first introduce the following set 
\begin{equation}\label{HOLCONT1}
{\cal D}^{cont}_{Hol} = \{\rho_{\langle m,k \rangle} \in {\cal D}_{Hol}: p(\rho_m) = 1- p(\rho_k)\}. 
\end{equation}
The elements of ${\cal D}^{cont}_{Hol}$ allow us to extend the notion of syntactic contradiction to ${\cal QC}^{Hol}_{AN}$ in the following way, 

\begin{definition}
{\rm An expression of the form ${\mathbb {AND}}_{Hol}(\rho_{\langle m,k \rangle})$ is said to be an {\it holistic contradiction} whenever   $\rho_{\langle m,k \rangle} \in {\cal D}^{cont}_{Hol}$.
}
\end{definition}

We note that an holistic contradiction can be characterized by a special value of $p({\mathbb {AND}}_{Hol}(\rho_{\langle m,k \rangle}))$, because:
\begin{equation}\label{EQCONTHOL}
\rho_{\langle m,k \rangle} \in {\cal D}^{cont}_{Hol} \hspace{0.3cm}{\mbox iff}\hspace{0.3cm}p({\mathbb {AND}}_{Hol}(\rho_{\langle m,k \rangle})) = \delta^{m,k}(\rho).
\end{equation}
In fact, by Theorem \ref{error} we have that:
\begin{eqnarray*}
p(\rho_m) = 1- p(\rho_k)& iff & p({\mathbb {AND}}^{m,k}_{Hol}(\rho)) + \beta^{m,k}(\rho) = 1- p({\mathbb {AND}}^{m,k}_{Hol}(\rho)) - \gamma^{m,k}(\rho) \\
& iff & p({\mathbb {AND}}^{m,k}_{Hol}(\rho) = 1- p({\mathbb {AND}}^{m,k}_{Hol}(\rho)) - \gamma^{m,k}(\rho) - \beta^{m,k}(\rho)\\  
& iff & p({\mathbb {AND}}^{m,k}_{Hol}(\rho)) = \delta^{m,k}(\rho).
\end{eqnarray*}

In other words, the notion of holistic contradiction is completely determinate by the elements of  ${\cal D}^{cont}_{Hol}$. For this reason if $\rho_{\langle m,k \rangle} \in {\cal D}^{cont}_{Hol}$, $\rho_{\langle m,k \rangle}$, it will be called as {\it holistically contradictory}.      

A version of Theorem \ref{error} for the elements of the set ${\cal D}^{cont}_{Hol}$ is established below.

\begin{theo} \label{EQCONTHOL}
Let $\rho_{\langle m,k \rangle} \in {\cal D}^{cont}_{Hol}$. Then:

\begin{enumerate}
\item
$p({\mathbb {AND}}_{Hol}(\rho_{\langle m,k \rangle}) = \frac{1- \beta^{m,k}(\rho) - \gamma^{m,k}(\rho)}{2}$,

\item
$ tr({\mathbb{T}}_p^{(m,k,1)}(\rho)) = \delta^{m,k}(\rho)^2 - \beta^{m,k}(\rho)  \gamma^{m,k}(\rho)$,

\item
$0\leq p({\mathbb {AND}}_{Hol}(\rho_{\langle m,k \rangle}) \leq \frac{1}{2}$,

\item
$p({\mathbb {AND}}_{Hol}(\rho_{\langle m,k \rangle}) = \frac{1}{2}$ iff $\beta^{m,k}(\rho) = \gamma^{m,k}(\rho) = 0$ iff $p(\rho_m) = p(\rho_k) = \frac{1}{2}$ iff $ tr({\mathbb{T}}_p^{(m,k,1)}(\rho)) = \frac{1}{4}$,

\item
$p({\mathbb {AND}}_{Hol}(\rho_{\langle m,k \rangle}) = 0$ iff $\beta^{m,k}(\rho) + \gamma^{m,k}(\rho) = 1$ iff $ tr({\mathbb{T}}_p^{(m,k,1)}(\rho)) =1- \beta^{m,k}(\rho)(1- \beta^{m,k}(\rho)) =1-  \gamma^{m,k}(\rho) (1- \gamma^{m,k}(\rho))$,

\end{enumerate}
\end{theo}

\begin{proof}
1) Since $p({\mathbb {AND}}_{Hol}(\rho_{\langle m,k \rangle}) = \delta^{m,k}(\rho)$, by Theorem \ref{error}-1, $1 = p({\mathbb {AND}}_{Hol}(\rho_{\langle m,k \rangle}) +  \beta^{m,k}(\rho) + \gamma^{m,k}(\rho) + \delta^{m,k}(\rho) = 2 \delta^{m,k}(\rho) + \beta^{m,k}(\rho) + \gamma^{m,k}(\rho)$. Thus, $p({\mathbb {AND}}_{Hol}(\rho_{\langle m,k \rangle}) = \frac{1- \beta^{m,k}(\rho) - \gamma^{m,k}(\rho)}{2}$. 

2) Immediate from Theorem \ref{error}-4 and Theorem \ref{EQCONTHOL}.

3) Since $0 \leq \beta^{m,k}(\rho) + \gamma^{m,k}(\rho) \leq 1$, by item 1,  $0\leq p({\mathbb {AND}}_{Hol}(\rho_{\langle m,k \rangle}) \leq \frac{1}{2}$.

4) By item 1, $p({\mathbb {AND}}_{Hol}(\rho_{\langle m,k \rangle}) = \frac{1}{2}$ iff $\beta^{m,k}(\rho) = \gamma^{m,k}(\rho) = 0$ iff $p(\rho_m) = \delta^{m,k}(\rho) = 1- p(\rho_k) = 1 - \delta^{m,k}(\rho) = \frac{1}{2}$ iff  $tr({\mathbb{T}}_p^{(m,k,1)}(\rho)) = \frac{1}{2}^2 - 0$  

5) Immediate from item 1, item 2 and Theorem \ref{EQCONTHOL}.

\qed
\end{proof}

\begin{example}[Werner states and syntactic contradiction]
{\rm
Werner states provide an interesting example of syntactical contradiction when a bipartition is considered. Werner states, originally introduced in \cite{WER} for two particles to distinguish between classical correlation and the Bell inequality satisfiability, have many interests for their applications in quantum information theory. Examples of this, are entanglement teleportation via Werner states \cite{LEEKIM}, the study of deterministic purification \cite{SHORT}, etc.

\begin{definition}\label{Werner}
{\rm Let us consider a Hilbert space ${\cal H}\otimes {\cal H}$ such that $dim({\cal H}) = n$. A {\it Werner state} on ${\cal H}\otimes {\cal H}$ is a density operator $\rho$ such that, for any $n$-dimensional unitary operator $U$, $$\rho=(U\otimes U)\rho(U^\dagger \otimes U^\dagger).$$
}
\end{definition}

We can express Werner states as a linear combination of the identity and ${\it SWAP}$ operators \cite[$\S$ 6.4.3]{HEIZ}: 

\begin{equation}\label{par}\rho = \rho_w^{(n^2)} = \frac{n+1-2w}{n(n^2-1)}I^{(n^2)}-\frac{n+1-2w n}{n(n^2-1)}{\it SWAP}^{(n^2)}
\end{equation} 
where $w\in [0,1]$ and ${\it SWAP}^{(n^2)}= \sum_{i,j}\ket{\psi_i} {\bra{\psi_j}} \otimes \ket{\psi_j} {\bra{\psi_i}} $ with $\ket{\psi_i}$ and $\ket {\psi_j}$ vectors of the standard $n$-dimensional computational basis.

Let us consider the Werner state $\rho^{(2^{2n})}_{w}$ in $\otimes ^{n+n}\mathbb C^2$. Then, we can prove that\footnote{For more technical details, see  {\rm \cite[Proposition 5.3]{FS}}.}:

\begin{enumerate}
\item $p({\mathbb {AND}}_{Hol}({\rho^{(2^{2n})}_w}_{\langle 2^n,2^n \rangle}))=\frac{2^{2n}+2^n(2w-1)-2}{4(2^{2n}-1)}$,

\item $p({\rho^{(2^{2n})}_w}_n)=\frac{1}{2}$, where ${\rho^{(2^{2n})}_w}_n$ is the partial trace of $\rho^{(2^{2n})}_{w}$ with respect to the subspace $\otimes^n\mathbb C^2$,

\item $tr({\mathbb{T}}_p^{(2^n,2^n,1)}({\bf M}(\rho^{(2^{2n})}_{w}) \otimes P_0)) =\frac{w 2^{n+1}-2^n-1}{4(2^{2n}-1)}$.

\end{enumerate}

By item 2 and by Eq.(\ref{HOLCONT1}) it can be proved that the Werner state ${\rho^{(2^{2n})}_{w}}_{\langle 2^n,2^n \rangle}$ is a syntactic contradiction for each $n \in {\mathbb {N}}$ and for any $w\in[0,1]$.

Figure 1 allows us to see the behavior of the Werner state $\rho^{(2^2)}_{w}$ as a syntactic contradiction taking into account the contribution of each parameter that defines the probability value $p({\mathbb {AND}}_{Hol}({\rho^{2^2}_{w}}_{\langle 2,2 \rangle)})$.

\begin{figure}\label{WERCONT}
\centering
\includegraphics[scale=0.25]{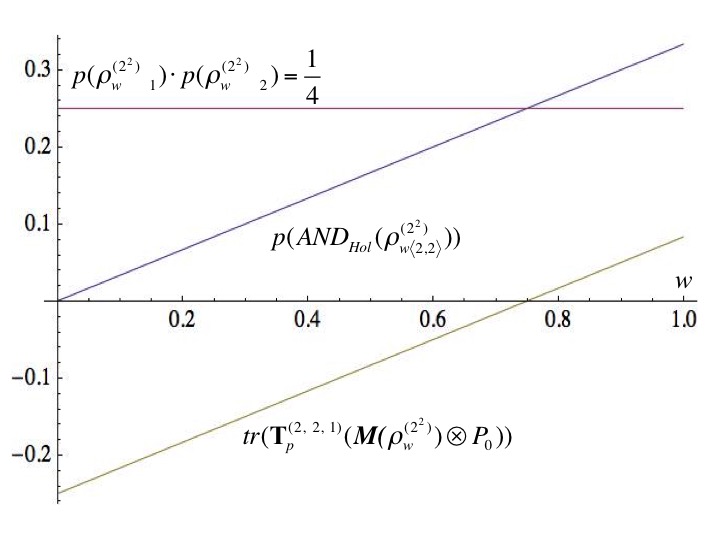}
\caption{Werner as holistic contraddiction and incidence of \hspace{0.1cm} $tr({\mathbb{T}}_p^{(2^n,2^n,1)}({\bf M}(\rho^{(2^{2n})}_{w}) \otimes P_0)).$}\label{nwerner}
\end{figure}

}
\end{example}

\section*{Conclusions}

In this work two semantical extensions of classical logic based on quantum computation with mixed states was investigated: the first, named ${\cal QC}_{AN}$, is a fuzzy type extension, while the second, named ${\cal QC}^{Hol}_{AN}$, is an improving of ${\cal QC}_{AN}$, where also holistic characteristics of bipartite quantum systems are considered.
Both extensions are conceived from logical connectives where natural interpretations are instances of Toffoli quantum gate acting on mixed states.

Formal aspects of these new logical systems were detailed in the paper, and they naturally suggest many interesting open questions and further developments in connection with different research areas. From the perspective of the philosophy of logic, ${\cal QC}_{AN}$ motivates new interpretations of fuzzy connectives in quantum computation. More precisely, some fuzzy logical systems, besides being related to the approximate reasoning or many-valued reasoning \cite{CDM}, admit quantum probabilistic interpretations associated to quantum circuits also. In the fuzzy context, notions like truth, tautology and logical consequences, may have another interpretation in the quantum computational framework. Technically speaking, ${\cal QC}_{AN}$ provides a good probabilistic description of circuits built on Toffoli quantum gates playing a similar role to classical logic in the digital techniques context. ${\cal QC}_{AN}$ deals with the ideal case where only factorizable states are taken into account.
The holistic extension ${\cal QC}^{Hol}_{AN}$, instead, is able to describe combinational aspects of Toffoli quantum gate in a more general realistic way. As we have seen in Section 5, ${\cal QC}^{Hol}_{AN}$ is strongly related to the fuzzy systems that defines ${\cal QC}_{AN}$. Further, this logical system provides an interesting connection between
some holistic features arising from non-factorizable bipartite states and standard fuzzy logic. 
From an epistemological point of view ${\cal QC}_{AN}$ and ${\cal QC}^{Hol}_{AN}$ can be considered as probabilistic type logics defining new kinds of quantum logic.

From an implementative perspective, these logical extensions can be very useful in quantum computing since the fuzzy content of ${\cal QC}_{AN}$ and ${\cal QC}^{Hol}_{AN}$ could be               
specially applied in fuzzy control \cite{DRIA}, allowing to model the so called Pelc's game \cite{MONT1} (a probabilistic variant of Ulam's game). 
It also suggests further developments in the study of error-correcting codes in the context of quantum computation.

\end{document}